\documentclass[preprint,12pt]{elsarticle}

\usepackage{latexsym,amssymb,amsmath,amsthm,color}
\usepackage{geometry}
\usepackage{hyperref}
\usepackage{url}
\usepackage{graphicx}
\usepackage{mathdots}
\usepackage[utf8]{inputenc}
\usepackage[numbers]{natbib}
\usepackage{xcolor}
\usepackage{multirow}
\usepackage{tikz}
\DeclareMathOperator{\diag}{diag}

\newcommand{\ie}{{\it i.e.}}
\newcommand{\F}{\mathbb{F}}
\newcommand{\N}{{\mathbb N}}
\newcommand{\rank}    {{\rm rank }}
\newcommand{\C}{\mathcal{C}}

\newcommand{\drfree}{\mbox{$d_{\mbox{\rm\tiny SR}}$}}
\newcommand{\dcrfree}{\mbox{$d^j_{\mbox{\rm\tiny SR}}$}}  
\newcommand{\drank}{\mbox{$d_{\mbox{\rm\tiny R}}$}}
\newcommand{\dsumrank}{\mbox{$d_{\mbox{\rm\tiny SR}}$}}

\newcommand{\im}       {{\rm im}}

\newcommand{\sgn}{\rm sgn}

\newcommand{\XXX}[2]{
  \begin{bmatrix}
    #1_0 & #1_1 &  \cdots & #1_#2 \\
     & #1_0  & \cdots & #1_{#2-1} \\
     &  & \ddots & \vdots \\
     &  &  & #1_0 \\
  \end{bmatrix}}

\theoremstyle{definition}
\newtheorem{definition}{Definition}

\theoremstyle{plain}
\newtheorem{theorem}{Theorem}

\newtheorem{lemma}[definition]{Lemma}
\newtheorem{obs}{Observation}

\usepackage{amssymb}

\begin{document}

\begin{frontmatter}



\title{Systematic Maximum Sum Rank Codes}
\author{Paulo Almeida, Umberto Mart{\'i}nez-Pe\~{n}as and Diego Napp }

\address{ Paulo Almeida
Dept.\ of Mathematics, University of Aveiro, Portugal, palmeida@ua.pt \\
Umberto Mart{\'i}nez-Pe\~{n}as, Dept.\ of Electrical \& Computer Engineering,
University of Toronto, Canadaumberto@ece.utoronto.ca\\
Diego Napp,  Dept.\ of Mathematics, University of  Alicante, Spain, diego.napp@ua.es \\
}

\begin{abstract}
In the last decade there has been a great interest in extending results for codes equipped with the Hamming metric to analogous results for codes endowed with the rank metric. This work follows this thread of research and studies the characterization of systematic generator matrices (encoders) of codes with maximum rank distance. In the context of Hamming distance these codes are the so-called Maximum Distance Separable (MDS) codes and systematic encoders have been fully investigated. In this paper we investigate the algebraic properties and representation of encoders in systematic form of Maximum Rank Distance (MRD) codes and Maximum Sum Rank Distance (MSRD) codes. We address both block codes and convolutional codes separately and present necessary and sufficient conditions for an  encoder in systematic form to generate a code with maximum (sum) rank distance. These characterizations are given in terms of certain matrices that must be superregular in a extension field and that preserve superregularity after some transformations performed over the base field. We conclude the work presenting some examples of Maximum Sum Rank convolutional codes over small fields. For the given parameters the examples obtained are over smaller fields than the examples obtained by other authors.
\end{abstract}


\end{frontmatter}

\section{Introduction} \label{sec intro}

Maximum Distance Separable (MDS) codes are block codes whose minimum Hamming distance attains the Singleton bound. In the linear case, they are characterized by having a generator matrix $G\in \F_q^{k \times n}$ whose full size $k \times k$ minors are all nonzero, where $k<n$ and $\F_q$ is a finite field. If $ G $ is in systematic form, \ie, $ G = [I_k \ P] $, where $I_k$ is the identity matrix of size $k$, then, the MDS property can be characterized only in terms of the matrix $P$, namely, all minors (of any size) of $P$ must be nonzero. One well-known example of this class of matrices is the set of generalized Cauchy matrices \cite{BlaumRoth}, which correspond to systematic generator matrices of Reed-Solomon codes \cite{ma77}.

In the last decade, rank-metric codes have been a very active area of research due to their wide range of applications in reliable and secure linear network coding \cite{KoetKsch08}, postquantum cryptography \cite{Horlemann-Trautmann2018}, local repair in distributed storage \cite{rawat} and space-time coding \cite{space-time-kumar}. Block codes attaining the Singleton bound for the rank metric are called maximum rank distance (MRD) codes. Unfortunately, if one represents such codes as subsets of $ \F_{q^M}^n $ with ranks defined over $ \F_q $, then all known MRD constructions (e.g., \cite{Gabidulin85}) are only decodable in superlinear time in $ n $ over $ \F_{q^M} $, where $ M \geq n $ (thus $ q^M $ is exponential in $ n $), and achieving linear-time decoding is already an extremely hard problem even in the Hamming metric. Thus it seems that MRD block codes will hardly ever be practical in real applications.

Recently, the sum-rank metric, which simultaneously extends the Hamming and rank metrics, has gained interest in the area. It was implicitly considered for multiple fading blocks in space-time coding \cite[Sec. III]{space-time-kumar}, and then formally introduced for multishot network coding \cite{ma16,secure-multishot,NappPintoVetRos2017,WachStinSido15}. The sum-rank metric, in the form of column rank distances, is the natural metric for convolutional codes tailored for streaming over linearly coded networks \cite{NappPintoVetRos2017,ma16}. The problem of constructing convolutional codes with maximum column (Hamming or rank) distances has attracted a lot of attention in recent years \cite{al16,BaYt2018,gl03,HaOs2018,Lieb2018}. In the Hamming context, it was shown in \cite{gl03} that the construction of convolutional codes with optimal column distances boils down to the construction of lower (block) triangular Toeplitz superregular matrices (see the formal definition in Section \ref{sec:preliminaries}). Several results on superregular matrices have been recently presented in \cite{ al16,BaYt2018,HaOs2018,Lieb2018,NaRo2015}. However, over small fields only constructions of superregular matrices with small parameters have been presented, most of them found by computer search. General constructions have also been presented, but require unpractical large finite fields \cite{AlmeidaNappPinto2013,gl03}, e.g., doubly exponential.

On the other hand, block codes equipped with the sum-rank metric are also of interest for reliable and secure multishot network coding \cite{secure-multishot} and especially for local repair in distributed storage \cite{lrc-sum-rank}. This is due to linearized Reed-Solomon codes \cite{MARTINEZPENAS2018587}, which are the only known MSRD (maximum sum-rank distance) block codes with subexponential field sizes for decoding, in contrast with MRD block codes which are practical only for moderate parameters (see \cite[Table I]{secure-multishot} and \cite[Section VI]{lrc-sum-rank}).

In this work, we give sufficient and necessary conditions for linear block and convolutional codes to be  MSRD only in terms of the matrix $ P\in \F_{q^M}^{k \times (n-k)} $ (resp. $ P(D) \in \F_{q^M}^{k \times (n-k)}[D] $), when the generator matrix of the code is in systematic form, \ie, $[ I_k \ \ P]$ (resp. $[ I_k \ \ P(D) ]$). For consistency with the convolutional literature, we will use the terms encoder and generator matrix interchangeably. These conditions will require not only that all the (non-trivial) minors of $P$ (resp. $ P(D)$) are nonzero, but that they remain nonzero after some operations over the base field $\F_q$. Worth mentioning  is the thorough work in \cite{Neri18,Neri_Thesis} on systematic  encoders of block MRD codes and the family of generalized Gabidulin codes. This class of codes are the rank analogue of the generalized Reed-Solomon codes and their nonsystematic encoders are given by $s$-Moore matrices, the $q$-analogues of weighted Vandermonde matrices. We note that  MRD codes are $q$-analogues of MDS codes, but MSRD codes are not. 

The outline of this paper is as follows. In Section 2, we present fundamental preliminary results on the structure of rank codes, convolutional codes and superregular matrices. In Section 3, we address and present first a matrix characterization for a systematic block code to be MRD and MSRD. We then proceed to tackle the more involved characterization of systematic convolutional codes with optimal column rank distances. We also address the general case of a convolutional code that does not necessarily admit a systematic polynomial encoder. The last section is devoted to present concrete examples of optimal codes over relatively small field sizes which improve the existing examples in the literature.

\section{Preliminaries}\label{sec:preliminaries}

In this section, we present the setting and necessary results to address the problems in the remainder of the paper.

\subsection{Block codes}

A block code is simply a nonempty subset $ \C \subseteq \F_{q^M}^n $, which we will consider to be $ \F_{q^M} $-linear from now on. In case its dimension is $ k $, we call it an $ (n,k) $ code. Let $ M_n : \F_{q^M}^n \longrightarrow \F_q^{M \times n} $ denote the $ \F_q $-linear vector space isomorphism that expands every scalar in $ \F_{q^M} $ as a column vector in $ \F_q^M $, with respect to some basis. Then we may define the rank metric in $ \F_{q^M}^n $ by $ \drank(v,w) = \rank (v - w) $ (see \cite{Gabidulin85}), where $ \rank(v) = \rank(M_n(v)) $, for all $ v,w \in \mathbb{F}_{q^M}^n $. In this context, codes are sometimes considered as subsets of $ \F_q^{M \times n} $ to use matrix operations or to restrict the study to $ \F_q $-linear codes.

The rank metric admits a natural extension, called sum-rank metric. If we partition the code length $ n = n_1 + n_2 + \cdots + n_\ell $, then we may define the corresponding sum-rank metric (see \cite{space-time-kumar,MARTINEZPENAS2018587,gsrws,NappPintoVetRos2017,WachStinSido15}) as
\begin{equation}
\dsumrank (v,w) = \sum_{i=1}^\ell \drank (v_i, w_i) = \sum_{i=1}^\ell \rank(M_{n_i}(v_i - w_i)),
\label{eq sum-rank metric definition}
\end{equation}
for all $ v = (v_1, v_2, \ldots, v_\ell) \in \F_{q^M}^n $ and $ w = (w_1, w_2, \ldots, w_\ell) \in \F_{q^M}^n $, where $ v_i, w_i \in \mathbb{F}_{q^M}^{n_i} $, for $ i = 1,2, \ldots, \ell $.

The sum-rank metric measures the error and erasure correction capabilities of codes in multishot matrix-multiplicative channels, e.g., multishot network coding \cite{secure-multishot,WachStinSido15}, space-time coding with multiple fading blocks \cite{space-time-kumar} and local repair with multiple local groups \cite{lrc-sum-rank}.

Not surprisingly, the rank metric is recovered by setting $ \ell = 1 $, and the classical Hamming metric is recovered by taking $ n_1 = n_2 = \ldots = n_\ell = 1 $ (or $ \ell = n $).

The minimum rank distance of a code $ \C \subseteq \F_{q^M}^n $ is defined as
$$ \drank (\C) = \min \{ \drank (v,w) \mid v,w \in \C, v \neq w \} = \min \{ \rank (v) \mid v \in \C, v \neq 0 \}, $$
and analogously for the sum-rank metric $\dsumrank (\C)$. For an $ (n,k) $ code $ \C \subseteq \F_{q^M}^n $, a \textit{generator matrix} or \emph{encoder} is a full-rank matrix $G\in \F_{q^M}^{k \times n}$ such that
$${\mathcal C}  =  \im_{\F_{q^M}}G = \left\{ uG\ \ | \ \, u \in \F_{q^M}^{k}\right\}. $$
An encoder of $ \C $ is called systematic if it is of the form $ G = [I_k, P] $, for some matrix $ P \in \F_{q^M}^{k \times (n-k)} $, where $ I_k \in \F_{q^M}^{k \times k} $ denotes the identity matrix of size $ k $. Note that, by basic linear algebra, any block code has a unique systematic encoder (up to permutation of columns).

For an $(n,k)$ code over $ \C \subseteq \F_{q^M}^n $, the analogues of the Singleton bound are given by
\begin{equation}
\drank (\C) \leq \min \left\lbrace 1, \frac{M}{n} \right\rbrace (n - k) + 1 \quad \textrm{and} \quad \dsumrank (\C) \leq \min \left\lbrace 1, \frac{\ell M}{n} \right\rbrace (n - k) + 1,
\label{eq singleton bounds block analogues}
\end{equation}
for the rank and sum-rank metrics, respectively, being the second valid for $ n_1 = n_2 = \ldots = n_\ell $. The first was proven in \cite{Gabidulin85}, whereas the second was proven in \cite{lrc-sum-rank}.

The bounds in (\ref{eq singleton bounds block analogues}) are refinements of the information-theoretical classical bound
\begin{equation}
\drank (\C) \leq n - k + 1 \quad \textrm{and} \quad \dsumrank (\C) \leq n - k + 1,
\label{eq singleton bounds block}
\end{equation}
respectively. We will say that $ \C $ is MRD (maximum rank distance) or MSRD (maximum sum-rank distance) if it attains the bounds in (\ref{eq singleton bounds block}), respectively. The bounds in (\ref{eq singleton bounds block analogues}) imply that MRD codes (resp. MSRD codes) only exist if $ M \geq n $ (resp. $ M \geq n / \ell $). Note how (\ref{eq singleton bounds block analogues}) and (\ref{eq singleton bounds block}) coincide for the Hamming metric ($ \ell = n $) and the restriction on the extension degree $ M $ for the existence of MDS codes vanishes.

Gabidulin codes are a well-known class of MRD codes \cite{Gabidulin85} (see also \cite{tr16,OtOz2018}). It is worth noting that $ \drank(\C) \leq \dsumrank(\C) $, thus any MRD code is also MSRD. However, as noted above, MRD codes only exist if $ M \geq n $. Linearized Reed-Solomon codes \cite{MARTINEZPENAS2018587} are the only known MSRD codes with subexponential field sizes $ q^M $, and achieve the minimum extension degree $ M = n / \ell $ whenever $ n_1 = n_2 = \ldots = n_\ell $. Since the fastest decoding algorithms for MRD codes are superlinear in $ n $ over $ \mathbb{F}_{q^M} $, with $ M \geq n $, the known decoding algorithm for linearized Reed-Solomon codes that is quadratic in $ n $ over $ \F_{q^M} $, with $ M = n / \ell $, is more than a degree faster in XOR operations in multishot channels with $ \ell \gg 1 $ (see \cite[Table I]{secure-multishot}), which are the practical cases (e.g., \cite{space-time-kumar, lrc-sum-rank, secure-multishot}).


We shall provide necessary and sufficient conditions for a systematic encoder to be MRD. First, we recall a characterization for encoders not necessarily in systematic form. The result is a variant of \cite[Theorem 1]{Gabidulin85} and was explicitly presented in \cite[Theorem 3.2 and Corollary 3.3]{tr16} using the Bruhat decomposition for matrices. It is also an immediate consequence of the more general results \cite[Theorem 2]{umberto16} or \cite[Theorem 6]{umberto16}.

\begin{theorem}[\cite{tr16}]\label{th:gabidulin}
	Let $G\in \F_{q^M}^{k \times n}$ be an encoder of $\C$. Then, the following statements are equivalent.
	\begin{enumerate}
		\item $\C$ is MRD;
		\item all the full size minors of $GA$ are nonzero for all nonsingular matrices $A\in \F_q^{n \times n}$;
		\item  all the full size minors of $GU$ are nonzero for all nonsingular upper triangular matrices $U\in \F_q^{n \times n}$.
	\end{enumerate}
\end{theorem}

The main idea behind the proof of condition 3. implies condition 2. above is that every nonsingular matrix $A$ can be written as $A=VQU$ where $V$ and $U$ are upper triangular and $Q$ is a permutation matrix (see, for example, \cite{tyrty}). 
This decomposition is a consequence of the more general Bruhat decomposition for algebraic groups. Since we use this decomposition frequently throughout the paper, we include its proof below for the sake of completeness.

\begin{lemma}
Let $\mathbb{F}$ be a field, let $n$ be a positive integer and let $A\in\mathbb{F}^{n\times n}$ be a nonsingular matrix. Then there exist nonsingular upper triangular matrices $U,V\in\mathbb{F}^{n\times n}$ such that
\[A=VQU,\]
where $Q$ is a permutation matrix.
\end{lemma}
\begin{proof}
Let $B=P_nA$, where $P_n$ is the permutation matrix with ones in the anti-diagonal. Using Gauss-Jordan elimination we obtain $L^{-1}BU^{-1}=P$, so $B=LPU$, where $L\in\mathbb{F}^{n\times n}$ is a lower triangular matrix, $U\in\mathbb{F}^{n\times n}$ is an upper triangular matrix and $P$ is a permutation matrix. Since $|A|\neq 0$ then $L$ and $U$ are nonsingular. Let $V=P_nLP_n$, then clearly $V$ is nonsingular and upper triangular, also $L=P_nVP_n$. Therefore,
\[A=P_nB=P_nLPU=VQU, \]
where $Q=P_nP$.
\end{proof}

The following result extends Theorem \ref{th:gabidulin} to the sum-rank metric in general. This result was explicitly stated in a more general form in \cite[Proposition 7]{gsrws}, but was already observed in the proof of \cite[Theorem 3]{MARTINEZPENAS2018587}.

\begin{theorem}[\cite{gsrws}]\label{th MSRD characterization}
	Let $ n = n_1 + n_2 + \cdots + n_\ell $ be a code-length partition defining the sum-rank metric as in (\ref{eq sum-rank metric definition}). Let $G\in \F_{q^M}^{k \times n}$ be an encoder of an $ (n,k) $ code $\C \subseteq \F_{q^M}^n $. Then, $\C$ is MSRD if, and only if, all the full size minors of $GA$ are nonzero for all nonsingular block-diagonal matrices
	\begin{equation}
	A = {\rm diag}(A_1, A_2, \dots, A_\ell)= \left[
	\begin{array}{cccc}
	A_1 &  &  &  \\
	& A_2 &  &  \\
	&  & \ddots&  \\
	&  &  & A_\ell \\
	\end{array}
	\right] \in \F_q^{n \times n} ,
	\label{eq def block diagonal matrix}
	\end{equation}
	where $ A_i \in \F_q^{n_i \times n_i} $, for $ i = 1,2, \ldots, \ell $.
\end{theorem}

\subsection{Convolutional codes}

As opposed to block codes, convolutional encoders process an input stream of information bits over a shift register (with possible feedback) and converts it into a stream of transmitted bits. Therefore, they are very suitable for streaming applications \cite{gl03,ma15a,ma16,mc98}. In the sequel, we follow the module-theoretic approach to convolutional codes as it was described in \cite{gl03,Hutchinson2008,OuNaPiTo19,Lieb2018,ro99a1}. A \textit{convolutional code} ${\C}$ of rate $k/n$, or $ (n,k) $ convolutional code, is an $\F_{q^M}[D]$-submodule of $\F_{q^M}^n[D]$ of rank $k$. A full-row-rank matrix $G(D)\in \F_{q^M}^{k \times n}[D]$ with the property that
$${\mathcal C}  =  \im_{\F_{q^M}[D]}G(D) = \left\{ u(D)G(D)\ \ | \ \, u(D) \in \F_{q^M}^{k}[D]\right\} $$
 is called a \textit{generator matrix} or \emph{encoder} for $\C$. 
 When the code ${\mathcal C}$ admits an encoder in systematic form, \ie, $G(D)=[I_k \ \ P(D)]$, for some $ P(D) \in \F_{q^M}^{k \times (n-k)}[D] $, we say that ${\mathcal C}$ is systematic. Note that, as opposed to block codes, not all convolutional codes admit an encoder in systematic form, not even after column permutation. We will assume that the encoder  $G(D)$ is \textit{basic} \ie, it has a polynomial right inverse.

Write $v(D)=v_{0}+v_{1}D+\cdots+v_{\ell}D^{\ell} \in \F_{q^M}^n[D] $, and represent $G(D)$ as a matrix polynomial,
\[
G(D)=G_{0}+G_{1}D+\cdots+G_{m}D^{m},
\]
where $G_{m} \neq 0$ and $G_{i}=0$, for $i > m$. Then we call $m$ the memory of $G(D)$, and for $j=0,1,\dots, m$ we define the code's $ j $th \textit{truncated sliding generator matrix} as
\begin{equation}\label{eq:Gtrunc}
G_j^c=\XXX{G}{j}.
\end{equation}
The truncated codeword can then be represented as
\begin{eqnarray}
  v _{[0,j]}=(v_{0} ,    v_{1} ,    \ldots,  v_{j}) = (u_{0} ,    u_{1} ,    \ldots,  u_{j})  \ G_j^c \in \F_{q^M}^{(j+1) n} ,
\end{eqnarray}
where, for $ t = 0,1, \ldots, j $, it holds that
$$ v_t = \sum_{i=0}^{\ell} u_{t-i}G_i \in \F_{q^M}^n . $$

Observe that, representing truncated codewords in such a way, we may naturally endow them with the sum-rank metric as in the case of block codes. As usual, we will consider the rank metric in each block of $ n $ coordinates. Then, for $ v(D) = \sum_i v_i D^i\in \F_{q^M}^n[D] $ and $w(D) = \sum_i w_i D^i \in \F_{q^M}^n[D] $, we define as in the previous subsection,
$$ \dsumrank (v_{[0,j]} , w_{[0,j]} ) = \sum_{i=0}^j \drank (v_i, w_i) = \sum_{i=0}^j \rank(M_{n}(v_i - w_i)),$$
for $ j = 0 ,1, \ldots $ As usual, we may define $ \dsumrank(v(D), w(D)) = \lim_{j \rightarrow \infty} \dsumrank(v_{[0,j]} , w_{[0,j]} ) $. As in the block case, each term in the sum corresponds to a shot of a matrix-multiplicative channel, see \cite{NappPintoVetRos2017,ma16,WachStinSido15}.

For an $(n,k)$ convolutional code $\C \subseteq \F_{q^M}^n[D]$, we  define its \emph{free sum-rank distance} as
$$ \drfree(\mathcal{C})= \min \left\{ \sum_{i\geq 0}^{} \rank(v_i) \ \ | \ \
  v(D) \in \mathcal{C},
  v(D)\neq 0 \right\}, $$
and its $ j $th \emph{column rank distance}, for $ j = 0,1,\ldots $, as
$$ \dcrfree(\mathcal{C})= \min \left\{ \sum_{i=0}^{j} \rank(v_i) \ \ | \ \
  v(D) \in \mathcal{C},
  v_0\neq 0 \right\}. $$
As the Hamming distance is always larger than or equal to the rank distance, the following upper bound follows from \cite[Proposition 2.2]{gl03} (see also \cite[Lemma 1]{ma16}):
\begin{equation}\label{eq:col_bound}
\mbox{$d^j_{\mbox{\rm\tiny SR}}$}(\mathcal{C})  \leq (j+1)(n-k)+1.
\end{equation}

For a systematic $ (n,k) $ convolutional code ${\mathcal C}$ with encoder of memory $m$, it is easy to see that the number $(n-k)(m +1)+1$ is the maximum possible value for the free sum-rank (and Hamming) distance of $\C$.

The codes $\C$ having $\mbox{$d^j_{\mbox{\rm\tiny SR}}$}(\mathcal{C})  = (j+1)(n-k)+1$ for $j=0,1,\dots , m$ will be called {\em memory Maximum Sum Rank} convolutional codes ($m$-MSR) (see \cite{ma16}). In the context of Hamming metric these codes are closely related to the codes called {\em Optimum Distance Profile} \cite[p. 112]{jo99}. The following result gives necessary and sufficient conditions for a convolutional code to be $m$-MSR and it was presented in \cite[Theorem 3]{ma16}. First 
define

\begin{equation}\label{eq:diag}
A^*_{[0,j]}=\diag(A^*_0, A^*_1, \dots, A^*_j)= \left[
                              \begin{array}{cccc}
                                A^*_0 &  &  &  \\
                                 & A^*_1 &  &  \\
                                 &  & \ddots&  \\
                                 &  &  & A^*_j \\
                              \end{array}
                            \right],
\end{equation}
with $A^*_i\in \F_q^{n \times \rho_i}$  matrices, $\rho_i \in \N$ for $i=0,1,\dots, j$.

\begin{theorem}[\cite{ma16}]\label{th:canadianos}
For $0\leq i \leq j$, let $0 \leq \rho_i \leq n$ satisfy
\begin{equation}\label{eq:canadianos}
  \sum_{h=0}^{i} \rho_h \leq k(i+1),
\end{equation}
for all $i\leq j$ and with equality for $i=j$. The following are equivalent for any convolutional code:
\begin{enumerate}
  \item $\mbox{$d^j_{\mbox{\rm\tiny SR}}$}(\mathcal{C})  = (j+1)(n-k)+1$;
  \item for all full rank $A^*_{[0,j]}=\diag(A^*_0, A^*_1, \dots, A^*_j)$ constructed from full rank blocks $A^*_i \in \F_q^{n \times \rho_i}$ and $\rho_i$ that satisfy (\ref{eq:canadianos}), the product $G_j^c A^*_{[0,j]}$ is nonsingular.
\end{enumerate}
\end{theorem}

\subsection{Superregular matrices}


Superregular matrices have been a fundamental notion in coding theory as they can be used to construct systematic codes with optimal Hamming distance, both block and convolutional codes. Roughly speaking, this is due to the fact that a superregular matrix has the following property: Take any of its rows with Hamming weight, say  $d$. Then, any combination of this row with $t$ other rows yields a vector of Hamming weight $\geq d-t$,  see \cite[Theorem 3.1]{al16}. In this paper, we will show that the notion of superregular matrices can be also used to build codes with maximum rank or sum-rank distances. We will present later some constructions that involves finding a class of superregular matrices with entries in $\mathbb{F}_{q^M}$ that preserve the property of superregularity after some multiplication and addition of matrices in the base field $\mathbb{F}_{q}$. Next, we formally introduce the notion of superregular matrix.

Let $F= ( \mu_{i,j} )_{1 \leq i,j \leq m} \in \F_{q^M}^{m \times m} $, and let $S_m$ the symmetric group of order $m$. Recall that the determinant of $F$ is given by
\begin{equation}\label{deter}
|F|=\sum_{\sigma\in S_m}{\sgn(\sigma)}\mu_{1\sigma(1)}\cdots \mu_{m\sigma(m)},
\end{equation}
where the sign of the permutation $\sigma$, denoted by $\sgn(\sigma)$, is $1$ (resp. $-1$) if $\sigma$ can be written as product of an even (resp. odd) number of transpositions. A {\em trivial term} of the determinant is a term of (\ref{deter}), $\mu_{1\sigma(1)}\cdots \mu_{m\sigma(m)}$, equal to zero. If $F$ is a square submatrix of a matrix $B$, with entries in $\mathbb{F}_{q^M}$, and all the terms of the determinant of $F$ are trivial we say that $|F|$ is a \textit{trivial minor} of $B$. We say that $B$ is \textit{superregular} if all its non-trivial minors are different from zero. When a matrix has all its minors nonzero we call it \emph{full superregular}. These matrices have obviously all their entries nonzero and all minors are non-trivial. In that case the classical characterization of MDS block codes follows.

\begin{lemma}[\cite{ma77}]\label{lem:MDSblock}
  Let $\C= \im_{\mathbb{F}_{q^M}} G \subseteq \mathbb{F}_{q^M}^n $ be an $(n,k)$ block code. Then, $\C$ is MDS if, and only if, all $k \times k$ full size minors of $G$ are nonzero. If $G$ is in systematic form, \ie, $G =[I_k \ \ P]$  for some $P\in \mathbb{F}_{q^M}^{k \times (n-k)}$, then $\C$ is MDS if, and only if, $P$ is full superregular.
\end{lemma}

%

It is important to remark here that there exist several related notions of superregular matrices in the literature. Frequently, see for instance \cite{BlaumRoth}, a superregular matrix is defined to be full superregular, e.g., Cauchy matrices. In \cite{ma77}, several examples of triangular matrices were constructed in such a way that all submatrices inside this triangular configuration were nonsingular. However, all these notions do not apply to more general setting, such as the convolutional case, as in this context we need to consider submatrices that contain zeros. The more recent contributions \cite{gl03,HaOs2018,Hutchinson2008} consider the same notion of superregularity as in this paper (considering minors with zeros), but defined only for lower triangular Toeplitz matrices, see  \cite{HaOs2018} for examples of superregular matrices of size up to $5\times 5$. In \cite{BaYt2018} \emph{block} Toeplitz superregular matrices were considered for high rate convolutional codes (see also \cite{AlmeidaNappPinto2013,al16}). The advantage of the definition of superregularity considered here is that unifies all existing notions in the literature.

\section{Systematic encoders with maximum rank or maximum sum rank distance}


In this section, we extend the results that link MDS codes and superregular matrices to the context of rank-metric and sum-rank-metric codes, both block and convolutional. We will see that a code is MRD or MSRD if its systematic encoder with entries in $\F_{q^M}$ is superregular and remains superregular after some operations with matrices that have entries in the base field $\F_{q}$. We first treat the block case and then address the convolutional counterpart.

\subsection{Systematic Encoders of MRD codes}

If the encoder is given in systematic form one can derive a  characterization in terms only on the parity part of the encoder. This result was obtained independently  in \cite[Theorem 3.11]{Neri_Thesis}. The following proof will be useful for the proof of Theorem \ref{th MSRD characterization systematic} in the next subsection.

\begin{theorem}\label{th:syst_block_MRD}
Let $G= [I_k,  P] \in \F_{q^M}^{k \times n}$ be a systematic encoder of a rank metric block code $\C$. Then, the following statements are
equivalent.
\begin{enumerate}
  \item $\C$ is MRD;
  \item the matrix
  $$
   BP\widetilde{A} + C  \in \F_{q^M}^{k \times (n-k)}
  $$
is full superregular, for all $C\in \F_q^{k \times (n-k)}$ and for all nonsingular upper-triangular matrices $B\in \F_q^{k \times k}$ and $\widetilde{A}\in\F_q^{(n-k) \times (n-k)}$.
\end{enumerate}

\end{theorem}
\begin{proof}
$(2. \Rightarrow 1.): $ Firstly we will prove that if the matrix
$$ BP\widetilde{A} + C \in \F_{q^M}^{k \times (n-k)} $$
is full superregular, for all $C\in \F_q^{k \times (n-k)}$ and for all nonsingular upper-triangular matrices $B\in \F_q^{k \times k}$ and $\widetilde{A}\in\F_q^{(n-k) \times (n-k)} $, then all the full size minors of $GU$ are nonzero for all nonsingular upper triangular matrices $U\in \F_q^{n \times n}$. Hence the result follows from Theorem \ref{th:gabidulin}.

Suppose that there exists a $k \times k$ zero minor of $GU$ for a nonsingular upper-triangular $U\in \F_q^{n \times n}$. Write
  $$
  U= \left[
       \begin{array}{cc}
         \widetilde{C} &  \widehat{C}\\
         0 &  \widetilde{A} \\
       \end{array}
     \right],
  $$
with $\widetilde{C}\in \F_q^{k \times k}$ and   $\widetilde{A}\in \F_q^{(n-k) \times (n-k)}$ nonsingular upper-triangular matrices. Thus,
  $$ GU= [ \widetilde{C} \ \ \widehat{C} + P \widetilde{A} ].$$
Denote $B=(\widetilde{C})^{-1} \in \F_q^{k \times k}$, which is a nonsingular upper-triangular matrix. Thus, 
  $$
  BGU= [I_n \ \ C + BP  \widetilde{A} ],
  $$
where $C=B\widehat{C}$ is a matrix with entries in the base field $\F_q$. As the left multiplication of $GU$ by an invertible matrix $B$ does not change the zeroness of the full size minors of $GU$, we have that $ [ I_n \ \ C + BP  \widetilde{A} ]$ has a zero minor and by Lemma \ref{lem:MDSblock}, $C + BP  \widetilde{A} $ is not full superregular.

$(1. \Rightarrow 2.):$ Suppose that there are a matrix $C\in \F_q^{k \times (n-k)}$ and nonsingular upper-triangular matrices $B\in \F_q^{k \times k}$ and $\widetilde{A}\in\F_q^{(n-k) \times (n-k)}$ such that $C +  BP  \widetilde{A}$ is not full superregular. Then, there exists a $k \times k$ minor of $ [ I_ k \ \ C +  BP  \widetilde{A} ] $ equal to zero and therefore $B^{-1} [ I_ k \ \ C +  BP  \widetilde{A} ] $  has a full size zero minor. Thus,
   $$
   B^{-1} [ I_ k \ \ C +  BP  \widetilde{A} ] =[ B^{-1} \ \ B^{-1}C +  P  \widetilde{A} ] = [I_k \ \ P ] \left[
       \begin{array}{cc}
         B^{-1} & B^{-1}C \\
         0 & \widetilde{A}
       \end{array}
     \right].
     $$
Denote
\[  U= \left[
       \begin{array}{cc}
    B^{-1} & B^{-1}C \\
         0 & \widetilde{A}
       \end{array}
     \right].\]
It is straightforward to verify that $U$ is nonsingular and upper triangular and $GU$ has a full size zero minor. Taking into consideration the statements of Theorem \ref{th:gabidulin} this concludes the proof.
\end{proof}


\subsection{Systematic Encoders of MSRD codes}

In this subsection, we will give characterizations for MSRD systematic encoders. The main result (Theorem \ref{th MSRD characterization systematic}) is an extension of Theorem \ref{th:syst_block_MRD}, which can be recovered by setting $ \ell = 1 $, the case where the sum-rank metric becomes the rank metric.

It is very important to remark that, in contrast with block codes in the rank metric ($ \ell = 1 $) or Hamming metric ($ \ell = n $), the use of a non-trivial code-length partition $ n = n_1 + n_2 + \cdots + n_\ell $ implies that we no longer can assume that the identity matrix is placed in the first $ k $ coordinates. For full generality, we will need to consider arbitrary partitions $ k = k_1 + k_2 + \cdots + k_\ell $ of an information set. This will be transparent in the convolutional case, due to their polynomial nature and the fact that the identity matrix is a constant matrix (see Theorem \ref{th:main}).


\begin{theorem}\label{th MSRD characterization systematic}
Let $ n = n_1 + n_2 + \cdots + n_\ell $ be a code-length partition defining the sum-rank metric as in (\ref{eq sum-rank metric definition}). Let $ k = k_1 + k_2 + \cdots + k_\ell $ be a dimension partition, where $ 0 \leq k_i \leq n_i $, for $ i = 1,2, \ldots, \ell $. Finally, let
$$ G = [J_1, P_1, J_2, P_2, \ldots, J_\ell, P_\ell] \in \F_{q^M}^{k \times n} $$
be a systematic encoder of an $ (n,k) $ code $\C \subseteq \F_{q^M}^n $, where $ P_i \in \F_{q^M}^{k \times (n_i - k_i)} $ is arbitrary and where $ J_i \in \F_{q^M}^{k \times k_i} $ is zero everywhere except for the $ i $-th block of $ k_i $ rows, where it is the identity matrix of size $ k_i $, for $ i = 1,2, \ldots, \ell $. In other words, $ I_k = (J_1, J_2, \ldots, J_\ell)$. Denote $ P = [P_1, P_2, \ldots, P_\ell] \in \F_{q^M}^{k \times (n-k)} $. Then, the following are equivalent,
\begin{enumerate}
  \item $\C$ is MSRD;
  \item the matrix
$$ \left[ \begin{array}{cccc}
  B_1 &  &  &  \\
  & B_2 &  &  \\
  &  & \ddots&  \\
  &  &  & B_\ell \\
 \end{array} \right] P \left[ \begin{array}{cccc}
 \widetilde{A}_1 &  &  &  \\
 & \widetilde{A}_2 &  &  \\
  &  & \ddots&  \\
  &  &  & \widetilde{A}_\ell \\
 \end{array} \right] + \left[ \begin{array}{cccc}
  C_1 &  &  &  \\
  & C_2 &  &  \\
  &  & \ddots&  \\
  &  &  & C_\ell \\
 \end{array} \right] \in \F_{q^M}^{k \times (n-k)} $$
is full superregular, for all matrices $ C_i \in \F_q^{k_i \times (n_i-k_i)}$ and for all nonsingular upper-triangular matrices $B\in \F_q^{k_i \times k_i}$ and $\widetilde{A}\in\F_q^{(n_i-k_i) \times (n_i-k_i)}$, for $ i = 1,2, \ldots, \ell $.
\end{enumerate}


\end{theorem}
\begin{proof}
We may take the matrices $ A_i $ in Theorem \ref{th MSRD characterization} to be upper triangular. The arguments in this proof will be an extension of those in the proof of Theorem \ref{th:syst_block_MRD}. However, we will need to be careful regarding the partition of the information set, which is not an issue in the cases $ \ell = 1 $ or $ \ell = n $, as explained at the beginning of this subsection.

Assume that there exists a $ k \times k $ zero minor in $ GA $, for a nonsingular block-diagonal matrix $ A $ as in (\ref{eq def block diagonal matrix}), where $ A_i $ is upper triangular, for $ i = 1,2, \ldots, \ell $. Write
  $$
  A_i = \left[
       \begin{array}{cc}
         \widetilde{C}_i &  \widehat{C}_i \\
         0 & \widetilde{A}_i \\
       \end{array}
     \right] \in \F_q^{n_i \times n_i} ,
  $$
where $ \widetilde{C}_i \in \F_q^{k_i \times k_i} $, $ \widetilde{A}_i \in \F_q^{(n_i - k_i) \times (n_i - k_i)} $ and $\widehat{C}_i \in \F_q^{k_i \times (n_i - k_i)} $, for $ i = 1,2, \ldots, \ell $ (notice that if for some $i$, $k_i=0$ then both matrices  $ \widetilde{C}_i$ and $ \widehat{C}_i$ do not exist and if $k_i=n_i$ then both matrices $ \widetilde{A}_i$ and  $ \widehat{C}_i$ do not exist). Now we have that
$$ GA = [J_1, P_1, J_2, P_2, \ldots, J_\ell, P_\ell] \left[ \begin{array}{cccc}
       \begin{array}{cc}
         \widetilde{C}_1 &  \widehat{C}_1 \\
         & \widetilde{A}_1 \\
       \end{array} &  &  &  \\
  & \begin{array}{cc}
         \widetilde{C}_2 &  \widehat{C}_2 \\
         & \widetilde{A}_2 \\
       \end{array} &  &  \\
  &  & \ddots&  \\
  &  &  &  \begin{array}{cc}
         \widetilde{C}_\ell &  \widehat{C}_\ell \\
         & \widetilde{A}_\ell \\
       \end{array} \\
 \end{array} \right] $$

$$ = \left[ \begin{array}{ccccccc}
\widetilde{C}_1 &  \widehat{C}_1 & & & & & \\
& & \widetilde{C}_2 &  \widehat{C}_2 & & & \\
& & & & \ddots & & \\
& & & & & \widetilde{C}_\ell &  \widehat{C}_\ell
\end{array} \right] + P \left[ \begin{array}{ccccccc}
0 & \widetilde{A}_1 & & & & & \\
 & & 0 &\widetilde{A}_2 & & & \\
 & & & & \ddots & & \\
  & & & & & 0 & \widetilde{A}_\ell
\end{array} \right]. $$
Denote $ B = {\rm diag}((\widetilde{C}_1)^{-1}, (\widetilde{C}_2)^{-1}, \ldots, (\widetilde{C}_\ell)^{-1}) \in \F_q^{k \times k} $, which is nonsingular and upper triangular. Then the reader can check that
$$ B G A = [J_1, T_1, J_2, T_2, \ldots, J_\ell, T_\ell ] \in \F_{q^M}^{k \times n}, $$
where again $ I_k = [J_1, J_2, \ldots, J_\ell ] $, but now
$$ T = [T_1, T_2, \ldots, T_\ell] = B P \left[ \begin{array}{cccc}
 \widetilde{A}_1 &  &  &  \\
 & \widetilde{A}_2 &  &  \\
  &  & \ddots&  \\
  &  &  & \widetilde{A}_\ell \\
 \end{array} \right] + \left[ \begin{array}{cccc}
  C_1 &  &  &  \\
  & C_2 &  &  \\
  &  & \ddots&  \\
  &  &  & C_\ell \\
 \end{array} \right] \in \F_{q^M}^{k \times (n-k)}, $$
where $ C_i = \widetilde{C}_i^{-1} \widehat{C}_i \in \F_q^{k_i \times (n_i - k_i)} $, for $ i = 1,2, \ldots, \ell $. As in the proof of Theorem \ref{th:syst_block_MRD}, it follows that $ T $ is not full superregular.

The reversed implication is proven as in Theorem \ref{th:syst_block_MRD}, again taking into account the partition $ k = k_1 + k_2 + \cdots + k_\ell $ and the block-diagonal nature of the corresponding matrices.
\end{proof}

It is immediate to see that Theorem \ref{th:syst_block_MRD} is the particular case of Theorem \ref{th MSRD characterization systematic} obtained by setting $ \ell = 1 $. A bit less trivial, but still easy, is to check that the classical characterization of MDS systematic encoders is recovered from Theorem \ref{th MSRD characterization systematic} by setting $ \ell = n $, or equivalently $ n_1 = n_2 = \ldots = n_\ell = 1 $. To that end, observe that the block-diagonal matrices $ {\rm diag}(B_1, B_2, \ldots, B_\ell) $ and $ {\rm diag}(\widetilde{A}_1, \widetilde{A}_2, \ldots, \widetilde{A}_\ell) $ are nothing but nonsingular diagonal matrices, i.e., monomial matrices. As for $ C ={\rm diag}(C_0, \dots, C_m) $ we can assume without loss of generality that the partition of $k$ is $k_i=1$ for $i\leq k$ and $k_i=0$ for $k<i\leq n$. So for any $i=1, \dots, n$ we have $n_i-k_i=0$ or $k_i=0$, so $\widehat{C}_i$ never exists, therefore $C$ does not exist either and we recover the classical characterization of MDS systematic encoders.

\subsection{Systematic Encoders of $m$-MSR convolutional codes}

Column distance is arguably the most fundamental distance measure for convolutional codes, \cite[pag. 109]{jo99}. A full characterization of polynomial encoders $G(D)$ that yield codes with optimal column Hamming distance was given in \cite{gl03} for general encoders and in \cite{Gabidulin88} when the encoder is in systematic form, see also \cite[Corollary 2.5]{gl03}. In this section we provide analogous characterizations in the context of rank metric convolutional codes. We start with a general result about full size minors.


\begin{theorem}\label{th:main}
  Let $\C$ be an $[n,k]$ convolutional code with memory $m$ and a systematic encoder $G(D)=[ I_k \ \ \ P(D) ]$, where $P(D)=\displaystyle \sum^{m}_{i=0}P_i D^i \in \F_{q^M}^{k \times (n-k)}[D]$, and let $0\leq j\leq m$. Given $\underline{A}_\ell\in \F_q^{(n-k) \times (n-k)}$, $B_\ell\in \F_q^{k \times k}$ and $C_\ell\in \F_q^{k \times (n-k)}$, with $\ell=0,1,\dots , j$, consider the matrix $T_j=T_j((\underline{A}_\ell,B_\ell,C_\ell)_{\ell=0}^j)$ defined by
    \begin{equation}\label{eq:Tj}
    \begin{split}
    T_j= \left [\begin{array}{cccc}
    B_0 &  &  &  \\
    & B_1 &  &  \\
    &  & \ddots&  \\
    &  &  & B_j \\
    \end{array} \right ] & \XXX{P}{j}\XXX{\underline{A}}{j}+\\
    & \hspace{-30pt}+\left [
    \begin{array}{cccc}
    C_0 &  &  &  \\
    & C_1 &  &  \\
    &  & \ddots&  \\
    &  &  & C_j \\
    \end{array}
    \right ]=  \left [
    \begin{array}{cccc}
    T_{0,0} & T_{0,1} & \cdots & T_{0,j} \\
    & T_{1,1} &  &  \vdots \\
    &  & \ddots& T_{j-1,j} \\
    &  &  & T_{j,j} \\
    \end{array}
    \right ].
    \end{split}
    \end{equation}
   The following statements are equivalent:
  \begin{enumerate}
    \item $\mbox{$d^j_{\mbox{\rm\tiny SR}}$}(\mathcal{C})  = (j+1)(n-k)+1$;
   \item
Every square submatrix of $T_j$ with all its diagonal entries in the matrices $T_{s,t}$, where $s,t\in \{0,1,\dots, j\}$,
is nonsingular, for all $C_\ell\in \F_q^{k \times (n-k)}$ and all nonsingular upper triangular matrices $B_\ell\in \F_q^{k \times k}$,  $\underline{A}_\ell\in \F_q^{(n-k) \times (n-k)}$;
    \item $T_j$ is superregular for all $C_\ell\in \F_q^{k \times (n-k)}$ and all nonsingular upper triangular matrices $B_\ell\in \F_q^{k \times k}$,  $\underline{A}_\ell\in \F_q^{(n-k) \times (n-k)}$,  $\ell=0,1,\dots , j$.
  \end{enumerate}
\end{theorem}

\begin{proof}
($2. \Rightarrow 1.$) 
Let $ A^*_{[0,j]}=\diag( A^*_0,  A^*_1, \dots, A^*_j)$, where each $ A^*_i\in \F_q^{n \times \rho_i}$ is a full rank matrix and $\rho_i \in \N$ is such that $0 \leq \rho_i \leq n$ and

\begin{equation}\label{eq:cond_th3}
\sum_{h=0}^{i} \rho_h \leq k(i+1),
\end{equation}
for all $i = 0,1, \ldots, j$, with equality in (\ref{eq:cond_th3}) for $i=j$.
Since for each $i=0,1,\dots, j$, $A^*_i$ has full rank $\rho_i$, there exists a matrix $A'_i$ such that
\[\breve{A}_i= \left [ \begin{array}{cc}
A^*_i &  A'_i
\end{array}\right ]\in \F_q^{n \times n}\]
is nonsingular.
Suppose $G_j^c A^*_{[0,j]}$ is singular and
consider $\breve{A}_{[0,j]}= \diag(\breve{A}_0, \breve{A}_1, \dots, \breve{A}_j)$, then
\begin{equation*}
  \begin{split}
G_j^c \breve{A}_{[0,j]}=& G_j^c  \left[
\begin{array}{cccc}
A_0^* \ A'_0 &  &  &  \\
& A_1^* \ A'_1 &  &  \\
&  & \ddots&  \\
&  &  & A_j^* \ A'_j \\
\end{array}
\right],
\end{split}
\end{equation*}
has a null full size minor as $G_j^c A^*_{[0,j]} $ is a submatrix of it.

Now, by Bruhat decomposition, there exist nonsingular upper triangular matrices $A_i$ and $U_i$, such that $\breve{A}_i=A_iQ_iU_i$, where $Q_i$ is a permutation. Consider the nonsingular matrices $A_{[0,j]}=\diag(A_0, A_1, \dots, A_j)$, $U_{[0,j]}=\diag(U_0, U_1, \dots, U_j)$ and $Q_{[0,j]}=\diag(Q_0, Q_1, \dots, Q_j)$.
Then $\breve{A}_{[0,j]}= A_{[0,j]}Q_{[0,j]}U_{[0,j]}$. Since
\[G_j^cA_{[0,j]}=G_j^c \breve{A}_{[0,j]}U_{[0,j]}^{-1}Q_{[0,j]}, \]
then $G_j^cA_{[0,j]}$ has a submatrix $\mathbb{M}$ such that $|\,\mathbb{M}\,|=0$. Moreover,  this matrix is built by selecting $\rho_i$ columns between the $(ik+1)$-th and the $(i+1)k$-th columns of $G_j^c A_{[0,j]}$, $i=0,1, \dots, j$ and the $\rho_i$ satisfies (\ref{eq:cond_th3}). Write
\begin{equation*}
\begin{split}
A_i = \left[ \begin{array}{cc}
\hat{A}_i & C^*_i   \\
& \underline{A}_i \\
\end{array}
\right],
\end{split}\end{equation*}
where $\hat{A}_i\in \F_q^{k \times k}$ and $\underline{A}_i\in \F_q^{(n-k) \times (n-k)}$ are nonsingular upper triangular matrices.


Therefore,
\begin{equation*}
\begin{split}
G_j^c A_{[0,j]}=& \left[
\begin{array}{cccc}
(I_k \ \ P_0) & (0 \ \ P_1)& \cdots  &  (0 \ \ P_j)\\
& (I_k \ \ P_0) &  &  \vdots \\
&  & \ddots&  \\
&  &  & (I_k \ \ P_0) \\
\end{array}
\right]
\left[
\begin{array}{cccc}
A_0 &  &  &  \\
& A_1 &  &  \\
&  & \ddots&  \\
&  &  & A_j \\
\end{array}
\right]\\
=&     \left[
\begin{array}{cccc}
(\hat{A}_0 \ \ C^*_0 + P_0\underline{A}_0)& (0 \ \  P_1\underline{A}_1) & \cdots  &  (0 \ \ P_j\underline{A}_j)\\
& (\hat{A}_1 \ \ C^*_1 + P_0\underline{A}_0) &  &  \vdots \\
&  & \ddots&  \\
&  &  & (\hat{A}_1 \ \ C_j^* + P_j\underline{A}_j) \\
\end{array}
\right].
\end{split}\end{equation*}

Let $B_i=(\hat{A}_i)^{-1}\in \F_q^{k \times k}$ and $B_{[0,j]}=\diag(B_0, B_1, \dots, B_j)$.
Then
$$
B_{[0,j]}G_j^c A_{[0,j]}=    \\
\left[
\begin{array}{cccc}
(I_k \ \ C_0 + B_0P_0\underline{A}_0) & (0 \ \  B_0 P_1\underline{A}_1) & \cdots  &  (0 \ \ B_0 P_j\underline{A}_j) \\
& (I_k \ \ C_1 + B_1P_0\underline{A}_1) &  &  \vdots \\
&  & \ddots&  \\
&  &  & (I_k \ \ C_j + B_jP_0\underline{A}_j) \\
\end{array}
\right],
$$
where $C_i= B_i C_i^*$. Let $\rho_i= \mathfrak{I}_i + \mathfrak{T}_i$ where $\mathfrak{T}_i$ corresponds to the number of columns of $B_{[0,j]}\mathbb{M}$ selected from the block columns starting with $B_0 P_i\underline{A}_i$ and $\mathfrak{I}_i$ corresponds to number of columns of $B_{[0,j]}\mathbb{M}$ selected from the block columns containing the identity matrix $I_k$ in the $i$-th block position.

After permutation of columns, which again do not change the zeroness of the full minors, we obtain
\begin{eqnarray}\label{eq:systematic2}
\left[
\begin{array}{c|c}
I_{k(j+1) } &  \begin{array}{cccc}
C_0 + B_0P_0\underline{A}_0 & B_0 P_1\underline{A}_1  &  &  B_0 P_j\underline{A}_j \\
& C_1 + B_1P_0\underline{A}_1 &  &  \\
&  & \ddots &  \\
&  &  & C_j + B_jP_0\underline{A}_j \\
\end{array}
\end{array}
\right].\\
\underbrace{\hspace*{9.3cm}}_{{\displaystyle =T_j}}\hspace*{0.4cm} \nonumber
\end{eqnarray}
We are going to prove that there exists a unique square submatrix of $T_j$, say $M$, such that $|\, M\,|=0$ and its diagonal entries are in the matrices $T_{s,t}$, where $s,t\in \{0,1,\dots, j\}$.

Let $\hat{\mathbb{M}}$ be the matrix that corresponds to $B_{[0,j]}\mathbb{M}$ after the change of columns, $S_j=\mathfrak{I}_0+\mathfrak{I}_1+\cdots + \mathfrak{I}_j$ and denote by $c_1,c_2, \dots, c_{S_j}$ the indices of the columns of $\hat{\mathbb{M}}$ selected in the first $k(j+1)$ positions of (\ref{eq:systematic2}), and $c_{S_j+1}, \dots, c_{(j+1)k}$ the indices from the remaining columns of $\hat{\mathbb{M}}$. Let $M$ be the square submatrix of $T_j$, built by selecting the columns indexed by $c_{S_j+1}, \dots, c_{(j+1)k}$ and the rows indexed in $\{ 1,2, \dots, (j+1)k\} \setminus \{ c_1,c_2, \dots, c_{S_j} \}$. Therefore,
\[0=|\,\hat{\mathbb{M}}\,|=\pm |\,M\,|\]
by the Laplace expansion over each of the first $S_j$ columns of $\hat{\mathbb{M}}$.

It remains to show that all the diagonal entries of $M$ are in the matrices $T_{s,t}$, where $s,t\in \{0,1,\dots, j\}$, which we will achieve using the index condition (\ref{eq:cond_th3}) on the submatrix $\hat{\mathbb{M}}$.
Recall that $\rho_i= \mathfrak{I}_i + \mathfrak{T}_i$ and write

\begin{equation}\label{Mgeral}
M=\left[\begin{tabular}{c|c|c|c|c}
$M_0$ & &  & &  \\\cline{1 - 1}
\multirow{4}{*}{$O_0$} & $M_1$ &  &  & \\ \cline{2 -1}
& \multirow{3}{*}{$O_1$} & &  & \\ 
 & & $\cdots$ & $M_{j-1}$ &  \\
 &  &  &  & $M_j$ \\\cline{4 - 1}
 &  &  & \multirow{1}{*}{$O_{j-1}$} &
\end{tabular}
\right],
\end{equation}
with $M_i$ having $\mathfrak{T}_i$ columns and $k(i+1)-\sum_{h=0}^{i}\mathfrak{I}_h$ rows, $i=0,1,\dots, j$. Having all the entries of the diagonal of $M$ in the matrices $T_{s,t}$ amounts to saying that the number of rows of each $M_i$ is larger or equal to $\sum_{h=0 }^{i}\mathfrak{T}_h$, \ie,
$$
k(i+1)-\sum_{h=0}^{i}\mathfrak{I}_h \geq \sum_{h=0 }^{i}\mathfrak{T}_h,
$$
for $i=0,1,\dots, j$. But these are exactly the conditions (\ref{eq:cond_th3}). This  shows ($2.\Rightarrow 1.$).

($1.\Rightarrow 2.$) For the converse, let $M$ be a square singular submatrix of $T_j$ of order $\nu$, where $\nu\leq\min\{(j+1)k,(j+1)(n-k)\}$, with its diagonal entries in the matrices $T_{s,t}$, where $s,t\in \{0,1,\dots, j\}$. Suppose that $M$ is formed by the columns $d_1, \dots, d_\nu$ and rows $c_1, \dots, c_\nu$ of $T_j$. Clearly, $M$ can be written in the form (\ref{Mgeral}). Let $\mathfrak{T}_i$ be the number of columns of the matrix $M_i$ and let $\mathfrak{I}_i$ be such that the number of rows in $M_i$ is $k(i+1)-\sum_{h=0}^{i}\mathfrak{I}_h$. The condition on the entries of $M$ implies that
$$
k(i+1)\geq \sum_{h=0}^{i}\mathfrak{I}_h+\sum_{h=0 }^{i}\mathfrak{T}_h,
$$

Let $\rho_i= \mathfrak{I}_i+\mathfrak{T}_i$. Then the $\rho_i$ satisfy the conditions  (\ref{eq:cond_th3}). Let $\hat{\mathbb{M}}$ be the submatrix of the matrix $[I_{(j+1)k}\mid T_j]$ formed by the columns indexed in $\{ 1,2, \dots, (j+1)k\} \setminus \{ c_1,c_2, \dots, c_\nu \}$ and the columns $d_1, \dots, d_\nu$ of $T_j$, and all of the $(j+1)k$ rows. Then, by the Laplace expansion over each of its columns $\{ 1,2, \dots, (j+1)k\} \setminus \{ c_1,c_2, \dots, c_\nu \}$
$|\,\hat{\mathbb{M}}\,|=\pm |\,M\,|=0$.

If we write $C_i^*= C_iB_i^{-1}$,  $\hat{A}_i=B_i^{-1}$,
\begin{equation*}
\begin{split}
A_i = \left[ \begin{array}{cc}
\hat{A}_i & C^*_i   \\
& \underline{A}_i \\
\end{array}
\right]
\end{split}\end{equation*}
and  $A_{[0,j]}=\diag(A_0, A_1, \dots, A_j)$, then after a permutation of columns, $\hat{\mathbb{M}}$ corresponds to a submatrix, say $\mathbb{M}$, of $G_j^c A_{[0,j]}$ satisfying conditions (\ref{eq:cond_th3}). It is easy to see that this submatrix of $G_j^c A_{[0,j]}$ is equal to $G_j^c A^*_{[0,j]}$ for $A^*_i\in \F_q^{n \times \rho_i}$ and $\rho_i$ satisfying condition (\ref{eq:cond_th3}). Hence, if $M$ is singular, then $G_j^c A^*_{[0,j]}$ is also singular, and by Theorem \ref{th:canadianos} statement 1. fails to hold. Therefore ($1.\Rightarrow 2.$).

The equivalence ($2. \Leftrightarrow 3.$) readily follows from the fact that if there is a zero in the diagonal, all the entries to the right and below are also zero and therefore the determinant of such matrix is trivially zero.
%
\end{proof}

\begin{obs}
	Notice that condition 2. implies that all the entries of the matrices $T_{s,t}$, where $s,t\in \{0,1,\dots, j\}$ are nonzero.
\end{obs}

\subsection{General Convolutional Encoders}

If the convolutional code is not systematic one can easily transform the sliding generator matrix of a nonsystematic encoder in order to apply the conditions of Theorem \ref{th:main}. This fact is straightforward but worth mentioning.
%

Let $G(D)=[S(D) \ \ Q(D)]$ be the generator matrix of $\C$, where
$$
S(D)= \sum^{m}_{i=0} S_i D^i \in \F_{q^M}^{k \times k}[D],
$$
$$
Q(D)= \sum^{m}_{i=0} Q_i D^i \in \F_{q^M}^{k \times (n-k)}[D].
$$
As $\C$ is basic, we can assume without loss of generality that $S_0=I_k$, the identity matrix of size $k$. Further, let
\begin{equation}\label{eq:Laurent_expension}
  S^{-1}(D)Q(D)=\sum_{i=0}^{\infty} P_i D^i \in \F^{k \times (n-k)}(\!(D)\!)
\end{equation}
be the Laurent expansion of $S^{-1}(D)Q(D)$ over the field $\F(\!(D)\!)$ of Laurent series.

It is easy to see that, after a column permutation, the sliding parity-check matrix $G_j^c$, $j=0,1,\dots m$ of $\mathcal{C}$ has the form
\begin{equation*}
  G_j^c=\left ( \begin{array}{ccc}
   \XXX{S}{j} & \vline &
   \XXX{Q}{j}  \end{array}
\right )
\end{equation*}
and using that $S_0=I_k$  we can left multiply $G_j^c$ by the inverse of the first block to obtain
 \begin{equation}\label{eq:nonsyst_tosyst}
  \left ( \begin{array}{ccc}
   I_{k(j+1)}& \vline &
  \XXX{P}{j}
                          \end{array}
\right )
\end{equation}

As these operations do not change the full size minors of $G_j^c$ one can use the representation (\ref{eq:nonsyst_tosyst}) and Theorem \ref{th:main} and check whether $\mbox{$d^j_{\mbox{\rm\tiny SR}}$}(\mathcal{C})  = (j+1)(n-k)+1$ or not. 

We have considered in this paper $j\leq m$ but we note that non-systematic convolutional codes have column distances that grow to a time instant that can be larger than $m$, namely, $L=\left\lfloor \frac{\delta}{k} \right\rfloor +  \left\lfloor \frac{\delta}{n-k} \right\rfloor$, where $\delta$ is the \textit{degree} of the convolutional code defined to be the maximum of the degrees of the determinants of the $k\times k$ sub-matrices of one, and  hence any, generator matrix of $\mathcal{C}$, see \cite{ro99a1} for details. 

\section{Reducing the field size of $m$-MSR codes}\label{constMRP}


In \cite{ma16} a general construction of $m$-MSR convolutional codes was presented. Unfortunately, we were unable to find a general construction of superregular matrices that satisfy the conditions of Theorem \ref{th:main}. We conjecture, based on many examples, that the superregular matrices proposed in \cite{AlmeidaNappPinto2013,al16,ma16}, and presented below, satisfy such conditions and therefore can be used to build systematic $m$-MSR convolutional codes, but we were unable to formally prove it.

In any case, the main problem of all these general constructions is that they require impractically large finite fields. For this reason, most of the optimal constructions of convolutional codes presented over finite fields of reasonable size are found via computer search and limited to small parameters, see for instance \cite{BaYt2018,gl03,HaOs2018,Hutchinson2008}. In this section we present concrete examples of superregular matrices, of given parameters and finite fields, that satisfy conditions of Theorem \ref{th:main} and therefore yield $m$-MSR convolutional codes. The examples presented in \cite{ma16} were built requiring that $G_j^c$ is superregular and remains superregular after certain operations over $\F_q$ whereas in our constructions we only need that the smaller matrix $P_j^c$ is superregular and remains superregular after certain operations over $\F_q$. Therefore, the field sizes obtained improve the ones presented in \cite{ma16}, as Table \ref{texam} shows. The examples are built using the following superregular matrices. 

\begin{table}
\begin{center}
\begin{tabular}{ccccc}
$[n,k,m]$ & Achievable & $\#$ non-trivial & $\#$ Matrices & Achievable \\
& field & minors & $\underline{A}_\ell$ and $B_\ell$ & field in \cite{ma16} \\ \hline
$[2,1,1]$ & $\F_{2^2}$ & $1$ & $1\times 1$ & $\F_{2^{5}}$\\
$[2,1,2]$ & $\F_{2^3}$ & $7$ & $1\times 1$ & $\F_{2^{7}}$\\
$[3,2,1]$ & $\F_{2^4}$ & $5$ & $1\times 4$ & \\
$[3,1,1]$ & $\F_{2^5}$ & $6$ & $4\times 1$ & \\
$[4,2,1]$ & $\F_{2^6}$ & $40$ & $4\times 4$ & $\F_{2^{11}}$ \\
$[3,2,2]$ & $\F_{2^7}$ & $42$ & $1\times 8$ & $\F_{2^{11}}$\\
$[3,1,2]$ & $\F_{2^9}$ & $42$ & $8\times 1$ & $\F_{2^{11}}$\\
$[4,2,2]$ & $\F_{2^{11}}$ & $529$ &  $8\times 8$ & \\
$[5,3,1]$ & $\F_{2^{11}}$ & $136$ &  $4\times 64$ & \\
$[5,2,1]$ & $\F_{2^{12}}$ & $136$ &  $64\times 4$ & \\
$[6,4,1]$ & $\F_{2^{13}}$ & $335$ &  $4\times 4096$ & \\
$[6,2,1]$ & $\F_{2^{14}}$ & $670$ &  $4096\times 4$ & \\
$[6,3,1]$ & $\F_{2^{18}}$ & $634$ &  $64\times 64$ &
\end{tabular}
\caption{Parameters of $m$-MSR convolutional codes obtained by computer search.}\label{texam}
\end{center}
\end{table}


Let $\alpha$ be a primitive element of a finite field $\mathbb F_{q^M}$ with $q^M$ elements and consider $G(D)=  [I_k \ \ P(D) ]$ with $P(D)=\sum^{m}_{i=0}P_i D^i$ where $P_i$, for $0\leq i\leq m$ is equal to
\begin{eqnarray}\label{eq:coeff_SR}
P_i=\left[ \!\!
  \begin{array}{cccc}
    \alpha^{[Ri]} & \alpha^{[Ri+1]} & \cdots & \alpha^{[Ri+n-k-1]} \\
    \alpha^{[Ri+1]} & \alpha^{[Ri+2]} & \cdots & \alpha^{[Ri+n-k]} \\
    \alpha^{[Ri+2]} & \alpha^{[Ri+3]} &    \cdots & \alpha^{[Ri+n-k+1]} \\
    \vdots & \vdots & \vdots \\
    \alpha^{[Ri+k-1]}  & \alpha^{[Ri+k]}  & \cdots & \alpha^{[Ri+n-2]} \\
  \end{array}
\!\! \right]\in \F_{q^M}^{k \times (n-k)},
\end{eqnarray}
where $R=\max\{k,n-k\}$ and where we use the notation $\alpha^{[j]}= \alpha^{q^j}$ to denote the $j$-th Frobenius power of $\alpha \in \mathbb F_{q^M}$.
The next matrix
\begin{equation}\label{eq:slidingP}
  P_j^c= \XXX{P}{j}
\end{equation}
is superregular for all $j\leq m$ if the field size is sufficiently large, see \cite{AlmeidaNappPinto2013} for details. For smaller fields it may not be superregular. Nevertheless, for the parameters $[n,k,m]$ and field size displayed in Table \ref{texam}, this matrix satisfies the conditions of Theorem \ref{th:main} and can be used to construct $m$-MSR convolutional codes.


Table \ref{texam} shows the achievable fields obtained by computation. For all possible matrices $\underline{A}_\ell$ and $B_\ell$ (with $\ell=0, \dots, m$), all the non-trivial minors of $T_m-C$, are not in the base field $\F_2$, where $C = {\rm diag}(C_0, \dots, C_m)$, which implies that all the non-trivial minors of $T_m$ are nonzero. For this reason we did not consider any matrix $C$ in our calculations and therefore optimize our algorithms.


Another possibility in the quest to find optimal constructions of convolutional codes over small fields is to relax the condition of maximum (sum) rank distance and use instead the notion of almost MRD, see \cite{DELACRUZ2018}. This is left for future research.

\section*{Acknowledgement}

The first author is partially supported by the Portuguese Foundation for Science and Technology (FCT-Funda\c{c}\~{a}o para a Ci\^{e}ncia e a Tecnologia), through CIDMA - Center for Research and Development in Mathematics and Applications, within project UID/MAT/04106/2019. The second listed author gratefully acknowledges the support from The Independent Research Fund Denmark (Grant No. DFF-7027-00053B). The third author is partially supported by the the Universitat d'Alacant (Grant No. VIGROB-287) and Generalitat Valenciana (Grant No. AICO/2017/128).


\bibliographystyle{elsarticle-harv}
\bibliography{biblio_com_tudo}

\end{document}